\documentclass[a4paper, 11pt]{article}
\usepackage{amsmath, amssymb,amscd, amsthm}
\usepackage[mathscr]{eucal}
\usepackage{graphics}
\usepackage{fullpage}
\usepackage{url}
\newcommand\cyr{%
 \renewcommand\rmdefault{wncyr}%
 \renewcommand\sfdefault{wncyss}%
 \renewcommand\encodingdefault{OT2}%
\normalfont\selectfont} \DeclareTextFontCommand{\textcyr}{\cyr}

\newtheorem{theorem}{Theorem}
\newtheorem{lemma}[theorem]{Lemma}

\newtheorem{definition}[theorem]{Definition}

\begin{document}

\title{\textbf{On the Complexity of Asynchronous Agreement Against Powerful Adversaries}}

\author{Allison Lewko \\ Microsoft Research \and Mark Lewko \\ UCLA}

\date{}

\maketitle

\begin{abstract}
We introduce new techniques for proving lower bounds on the running time of randomized algorithms for asynchronous agreement against powerful adversaries. In particular, we define a \emph{strongly adaptive adversary} that is computationally unbounded and has a limited ability to corrupt a dynamic subset of processors by erasing their memories. We demonstrate that the randomized agreement algorithms designed by Ben-Or and Bracha to tolerate crash or Byzantine failures in the asynchronous setting extend to defeat a strongly adaptive adversary. These algorithms have essentially perfect correctness and termination, but at the expense of exponential running time. In the case of the strongly adaptive adversary, we show that this dismally slow running time is \emph{inherent}: we prove that any algorithm with essentially perfect correctness and termination against the strongly adaptive adversary must have exponential running time. We additionally interpret this result as yielding an enhanced understanding of the tools needed to simultaneously achieving perfect correctness and termination as well as fast running time for randomized algorithms tolerating crash or Byzantine failures.
\end{abstract}

\section{Introduction}
Achieving agreement in a distributed system despite failures is a central problem in distributed computing. We consider a complete network of $n$ processors able to communicate with each other by passing messages. Initially, each processor has an input bit. The task is to design a failure-resilient protocol that allows all non-faulty processors to agree on an output value, with the restriction that it must be equal to at least one of their inputs (this rules out the trivial solution of having a constant decision value independent of the inputs). The difficulty of this problem depends heavily on several additional specifications that must be made. In particular, is communication synchronous or asynchronous? What kinds of failures should be tolerated? If the errors and/or scheduling are controlled by an adversary, what resources and information does the adversary have access to?

We will consider a very challenging setting of asynchronous communication where message scheduling is controlled by an adversary with unbounded computational power who is given unrestricted access to all message contents and internal states of all processors. The adversary will also be empowered to cause limited types and quantities of processor failures. In this work, we will consider two kinds of failures: crash failures, which cause a processor to quit without warning, as well as resetting failures, which we will define and motivate below.

In this setting, the elegant result of Fischer, Lynch, and Paterson \cite{FLP83} shows that it is already impossible to design a deterministic protocol for agreement that always terminates, even if the adversary is limited to causing at most \emph{one} crash failure. A common approach for tolerating this obstacle in practice is to use an algorithm that terminates as long as worst-case scheduling does not occur indefinitely. This is a property achieved by the well-known Paxos algorithm constructed by Lamport \cite{L98}.
Randomized algorithms provide a potential alternative. Quickly following the impossibility result, Ben-Or \cite{BO83} and Bracha \cite{B84} presented randomized algorithms terminating with probability one, even against such strong adversaries. These algorithms were intuitively structured, and Bracha's algorithm tolerated an optimal number of failures, namely allowing for $t$ processors to behave in an arbitrary malicious fashion, for any $t < \frac{n}{3}$. Also, Aguilera and Toueg \cite{AT12} have provided a new correctness proof for Ben-Or's randomized consensus algorithm when there are $< \frac{n}{2}$ crash failures.
However, for some settings of the initial input bits, the algorithms of \cite{BO83,B84} run for time that is exponential in $n$ (with high probability) when $t = \Omega (n)$.

The algorithms in \cite{BO83,B84} seem to provide even stronger failure resilience than is captured by the adversarial model employed. In particular, the original proofs of correctness rely only on the fact that at most $t$ processors are faulty \emph{at one time}, where the notion of ``time" must be defined in an appropriate (and perhaps subtle) way. This gives some hope for recovering from even more than $t < \frac{n}{3}$ failures over the course of long executions if individual processor faults are fleeting occurrences. In particular, one might suppose that faulty processors could be detected and fixed during the course of a protocol execution, thereby allowing for more total failures.

In order to more fully characterize the failure resilience provided by the basic algorithm underlying \cite{BO83,B84}, we define the notion of \emph{resetting failures}. A resetting failure at a processor results in loss of internal state: a processor that is reset is assumed to lose the entire contents of its memory (except for its initial input bit and its output bit). A resetting failure can model a processor that is detected to be faulty and has its memory reset in order for it to rejoin the protocol as a non-faulty processor. We define a strongly adaptive adversary who can cause up to $t$ resetting failures in a certain window of time, where our measure of time is appropriately linked to the events of the execution. (Some kind of linking is necessary to avoid allowing the adversary to always cause a failure at the processor currently taking a step in the execution, for example.)

We prove that a simple variant of the algorithms in \cite{BO83,B84} is indeed successful against such a strongly adaptive adversary (with probability one). Of course, this retains the exponentially slow running time. We then show that exponential slowness for $t  = \Omega(n)$ is inherent to any algorithm achieving success with probability one against this strongly adaptive adversary. This provides a rather complete understanding of what is achievable in the presence of adaptive resetting faults.

In contrast, the relatively recent algorithm of Kapron et. al. \cite{K08} runs very quickly (polylogarithmic time in $n$) and tolerates $t  < \left(\frac{1}{3} - \epsilon \right)$ non-adaptive Byzantine failures, but incurs a non-zero probability of non-termination or termination with invalid outputs. It is an interesting question to study to what extent the sacrifices made here (non-adaptivity, non-zero probability of incorrect output) are necessary to achieve fast running time. The algorithm in \cite{K08} works by iteratively dividing the processors into small ``committees" that can afford to run the slow algorithm of \cite{B84} to hold elections to select random smaller subsets of processors to continue into new committees. A single final committee is reached that, with $1 - o(1)$ probability, contains a suitably bounded percentage of faulty processors. This final committee runs the algorithm of \cite{B84} and informs the other processors of the result.

It is clear that this approach cannot be used against an adaptive adversary, who can simply wait for the final committee to be determined and then cause faults. This approach also seems to inherently incur non-zero probability of an invalid result, as there is always a nonzero chance that the final committee is composed entirely of faulty processors. With the goal of beginning a systematic study of what can be achieved without incurring these disadvantages, Lewko \cite{L11} previously proved that a class of algorithms generalizing Ben-Or and Bracha's algorithms in \cite{BO83,B84} cannot achieve subexponential running time against an adversary causing $t = \Omega(n)$ non-adaptive Byzantine failures. The class of algorithms was restricted in several ways, including a constant bound on the support size of all message distributions sampled by processors and a requirement for received messages from different processors to be treated symmetrically. (For a more detailed description of the algorithm class, see \cite{L11}.)

The techniques we introduce to prove the lower bound against strongly adaptive adversaries can be applied in this setting to yield an exponential lower bound on running time for a new class of algorithms tolerating $t = \Omega(n)$ crash failures. This class is incomparable to the class considered in \cite{L11}, and this result yields several new insights. Most notably, our lower bound technique can tolerate arbitrary use of randomness by the processors, allowing us to avoid requiring any restriction on the support size as in \cite{L11}. We also avoid any requirement of symmetry in how received messages are treated, and our class is more intuitively defined.

Concurrently with this work, King and Saia \cite{KS13} have discovered a Las Vegas polynomial-time algorithm tolerating adaptive Byzantine faults that falls outside the classes of algorithms considered here and in \cite{L11}. This implies a separation between what can achieved against the classical adaptive Byzantine adversary and the strongly adaptive adversary.

\paragraph{Our Techniques} To prove the exponential lower bound on running time, we rely crucially on a general probabilistic inequality of Talagrand \cite{T}, which roughly states that any product distribution $\Omega_1 \times \Omega_2 \times \ldots \times \Omega_n$ cannot put too much weight simultaneously on two sets $A$ and $B$ in $n$-dimensional space that are ``far apart." For our purposes, ``far apart" can be interpreted as having Hamming distance $\Omega(n)$. We also can interpolate this result: if some product distribution $\Omega_1 \times \ldots \times \Omega_n$ puts significant weight on $A$ and some other product distribution $\Pi_1 \times \ldots \times \Pi_n$ puts significant weight on $B$, then there is some mixed product distribution $\Omega_1 \times \ldots \times \Omega_i \times \Pi_{i+1} \times \ldots \times \Pi_n$ that puts small weight on each of $A$ and $B$.

To use these tools in order to prove a lower bound on running time, we define iterative pairs of sets in the joint state space of the $n$ processors that represent different levels of progress towards a final decision. By leveraging the capabilities of the strongly adaptive adversary (or later by leveraging the defining properties of the algorithm class), we prove that each of these pairs of sets is sufficiently separated in Hamming distance.
We then apply the probabilistic inequalities repeatedly as an execution travels through the state space of the $n$ processors, showing that for some initial setting of the inputs, the adversary can prevent the algorithm from making much progress in a given window of time with high probability. This ultimately yields our lower bound.

Our approach of leveraging general properties of product distributions in this iterative fashion represents a meaningful expansion of the suite of available tools for proving lower bounds in a distributed setting. In particular, there are essentially only a few core tools for proving lower bounds for randomized algorithms, and previous approaches do not achieve exponential lower bounds on running time when arbitrary amounts of randomness can be used. We consider our new techniques to be the main contribution of this work.
In the following subsection, we briefly survey prior lower bounds and other relevant work.

\subsection{Related Work}
The problem of reaching agreement despite faults was introduced by Pease, Shostak, and Lamport in \cite{PSL80}, who also proposed the Byzantine failure model in \cite{PSL82}. Since its introduction, the problem of fault-tolerant agreement has been widely studied in a variety of models. Several works have considered computationally bounded adversaries, a setting in which cryptographic tools can be employed (\cite{R83,T84,BG93,CR93,CKS05,N02}, for example).
In the synchronous communication setting, polylogarithmic round randomized protocols for Byzantine agreement against non-adaptive adversaries were obtained in \cite{KSSV06a,KSSV06b,BPV06,GPV06}. Recent work has focused on reducing the communication overhead of synchronous protocols \cite{KS09,KS10}.
 
Several lower bounds are also known. In addition to the impossibility of deterministic algorithms in the asynchronous setting mentioned above, there is a sharp lower bound of $t$ rounds for deterministic algorithms in the synchronous setting \cite{DS82}. This lower bound is proven by assembling a chain of executions where any two adjacent executions are indistinguishable to some non-faulty processor and the two ends of the chain represent different decision values. This basic strategy is adapted and expanded in \cite{L11} to yield a lower bound for a class of randomized algorithms, but this class inherently limits the amount of randomness used in choosing an individual message.

Polynomial lower bounds for randomized algorithms include the result of Bar-Joseph and Ben-Or \cite{BB98}, which proves a lower bound of $t/\sqrt{n\log n}$ on the number of expected rounds for a randomized synchronous protocol against an adversary who can adaptively choose to fail $t$ processors. 
Interestingly, their proof employs Schectman's theorem to analyze one-round coin flipping games, similar to our core use of Talagrand's inequality. 
However, they employ different techniques to build their analysis of multiple rounds. 
Another lower bound is due to Attiya and Censor \cite{AC08}, who show that for any integer $k$, the probability that a randomized Byzantine agreement algorithm does not terminate in $k(n-t)$ steps is at least $1/c^k$ for some constant $c$. Their technique involves constructing a chain of indistinguishable executions and bounding the termination probability in terms of the length of the chain. 
Aspnes \cite{A97} proves a lower bound of $\Omega (t/\log^2 t)$ on the expected number of local coin flips for asynchronous algorithms against adaptive adversaries that holds in either the shared memory or message passing model. This result is proven by establishing an extension of the techniques in \cite{FLP83} to a randomized setting. In the shared memory model, there are polynomial time randomized algorithms tolerating crash failures, and tight bounds on their total step complexity are proven by Attiya and Censor in \cite{AC08a}. This work also uses an analysis of product probability spaces, similarly to the proof in \cite{BB98}.

\section{Models and Definitions}
We let $n$ denote the total number of processors, and consider each processor to be endowed with a unique identity between 1 and $n$. We let $ 0 < t < n$ be a fixed positive integer (we let $t$ be arbitrary for the purposes of definition, but note that in our theorems below, we will take $t = cn$ for a suitably small positive constant $c$). We assume that each processor has its own source of random bits, and all of these sources are unbiased and independent. Each processor also has a fixed input bit, and a write-once output bit that is initially set to $\bot$. We work in a message-passing model, where any single processor can send a message to any other processor along a dedicated ``message channel," meaning that the recipient of a message will always correctly identify the sender.
We let $\mathcal{M}$ denote the space of all possible messages (this can be infinite). An element $m \in \mathcal{M}$ contains a sender identity, a receiver identity, and a string of bits interpreted as its contents.

We define the \emph{state} of a processor to include the current contents of its memory (note that this holds its identity, its input bit, and its output bit with current value $0$, $1$, or $\bot$). We let $\Sigma$ denote the set of possible processor states. An $n$-tuple of states, $\sigma \in \Sigma^n$ specifies a \emph{configuration} for the $n$ processors.

An \emph{algorithm} $\mathcal{A}$ is a collection of probability distributions on $\Sigma \times \mathcal{M}^n$, parameterized by $\Sigma \times \mathcal{M}$. In other words, an algorithm specifies how a processor should sample a new state and outgoing messages, depending on the current state as well as a just received message. The new state may contain updated memory contents (the output bit may or may not change). The new sent messages can depend on the freshly received message, the current memory of the processor, and freshly sampled random bits. We include $\emptyset \in \mathcal{M}$ to allow a processor to choose not to send a message.

We will adopt the usual notion of asynchrony and imagine that message delivery is controlled by an adversary. We will allow our adversary complete access to the current states of all the processors and the contents of all messages. We also allow our adversary unbounded computational power.

It is typical to define an execution as a sequence of steps, where each step consists of a processor (potentially) receiving a message, performing some local computation, and then possibly placing some outgoing messages into a ``message buffer." The adversary then controls the sequence of steps by deciding which processor will take the next step and what message (if any) that processor will receive.

To model an adversary able to crash up to $t$ processors, one can insist that in any infinite execution, all but at most $t$ processors take infinitely many steps and that every message sent to an infinitely stepping processor is eventually delivered. It is also common to consider a stronger Byzantine adversary, who instead has the power to corrupt the messages sent by up to $t$ processors. In this setting, we may require all processors to take infinitely many steps - but note that corrupted processors may simulate crashed processors by maliciously choosing not to send messages (changing a non-empty message $m$ to $\emptyset$ is considered a permissible corruption by the adversary). We note that corruption of messages allows an adversary to lie about the random coins sampled by a limited number of processors.

We will instead define a adversary who is able to ``reset" a \emph{changing} set of $\leq t$ processors. \emph{Resetting} a processor will correspond to erasing the contents of its memory, except for its input bit, its output bit, its processor identity, and a special counter that will increment each time a reset occurs. We assume that a processor keeps a local copy of the counter's value in its state, and hence will detect a reset when the local copy is erased and the real counter is non-zero. This mechanism of detection is just a book-keeping device, the key point is that we are assuming resets are events processors can internally detect (note that this strengthens our lower bound result).


We now consider executions expressed as sequences of more fine-grained steps between configurations, where we allow three distinct types of steps. A \emph{resetting step} will cause the memory of a specified processor to be reset. A \emph{receiving step} will deliver a message from the message buffer to its intended recipient. The recipient will then perform a local computation (perhaps sampling from some fresh local randomness). This will be the only kind of step that involves randomization.

Finally, a \emph{sending step} will allow a processor to place a set of new messages into the message buffer (this set may be empty if the processor chooses not to send anything). We assume that a single sending step represents a complete response to prior events, meaning that if a processor takes a sending step and then takes another sending step without taking any resetting or receiving steps in between, then the second sending step will have no effect - the state of the processor will remain unchanged and no new messages will be sent.
The adversary will control the order and nature of the steps.

Given a partial execution expressed as a finite sequence of such steps, we define its probability (with respect to a fixed algorithm $\mathcal{A}$) to be the product of the probabilities of each state change induced by a step, under the distributions specified by the algorithm. (This is assuming the initial configuration is valid.) Note that this will be zero if deterministic transitions are not followed, or if a step indicates delivery of a message that was never sent, etc.

Naturally, it would be impossible to make progress against an adversary allowed to reset processors arbitrarily. In particular, we must design a model that rules out the trivial case of an adversary that resets the memory of the receiving processor after every message delivery, as no algorithm can make progress under such adverse circumstances.
To ensure feasibility, we could limit the adversary to resetting at most $t$ processors throughout an execution. However, we can achieve success against even stronger adversaries. To specify an interesting such adversary, we make one key definition:

\begin{definition} An \textbf{acceptable window} is a consecutive segment of steps of the following form. First, all $n$ processors take sending steps. Then, for sets $S_1, S_2, \ldots, S_n \subseteq [n]$ all of size $\geq n-t$, a sequence of receiving steps follows that delivers to each processor $i$ the messages just sent to it from processors in the set $S_i$. Finally, a sequence of at most $t$ resetting steps occurs.
\end{definition}

The notion of an acceptable window is a formal unit of ``time" during an execution in which at most $t$ processors are faulty, and hence the other processors may not receive any messages from them. One could imagine adding a requirement that $i \in S_i$ for each $i$, as a processor can always safely wait to receive a message from itself, but this is unnecessary, as no resets occur between sending and receiving. This means that any information the processor could pass to itself through a message can instead be stored directly in the processor's state. Thus, adding a requirement that $i \in S_i$ would be superfluous here.

We define the \emph{Strongly Adaptive Adversary} to be an adversary allowed to reset processors and control message sending and receiving up to the constraint that any infinite execution is composed entirely of adjacent, disjoint acceptable windows. We observe that this adversary is incomparable to the usual Byzantine asynchronous adversary. Our strongly adaptive adversary has the additional power to erase processor memory, but it lacks the power to have corrupted processors ``lie" about their local random bits.

Our use of the phrase ``the strongly adaptive adversary" is a bit imprecise, since this technically constitutes a class of adversaries in the following sense. A single adversary should be thought of as a deterministic function that maps a partial execution to a next applicable step. Such a function need not be efficiently computable. We will consider the class of strongly adaptive adversaries to be the collection of individual adversaries that satisfy the requirement above to produce acceptable windows.

It may seem a bit unnatural to impose the constraint that an adversary should stick to acceptable windows, but we feel this captures the intuitive notion that the adversary should only corrupt $t$ processors ``at one time," as otherwise progress would be impossible. Furthermore, as our main result in this model is a lower bound, placing restrictions on the adversary strengthens our result.

We call a configuration \emph{reachable} if it occurs as the consequence of some partial execution with non-zero probability that is decomposable as a concatenation of acceptable windows. Note that the notion of reachability depends on the algorithm employed.

\begin{definition}
We say an algorithm $\mathcal{A}$ achieves \textbf{measure one correctness} against all strongly adaptive adversaries if any reachable configuration contains only agreeing or $\bot$ output bits, (in other words, one output bit being 0 and another being 1 is disallowed, but any assortment of 0's and $\bot$'s or any assortment of $1$'s and $\bot$'s is allowed). We also require that when an output is not $\bot$, it must agree with one of the inputs. This means that if all processors have inputs equal to 0, the decision cannot be 1, and vice versa.
\end{definition}

\begin{definition}
We say an algorithm $\mathcal{A}$ achieves \textbf{measure one termination} if any infinite execution (composed of acceptable windows) in which some processor taking an infinite sequence of sending and receiving steps never sets its decision bit has probability zero (we define the probability of an infinite execution to be the limit of the probabilities of its finite partial executions).
\end{definition}

In an asynchronous setting, defining the running time of an execution can be a bit subtle. One typical measure is to consider the length of the longest message chain before a decision is reached, where a message chain includes messages $m_1, m_2, \ldots, m_k$ such that $m_i$ is received by the sender of $m_{i+1}$, at some point prior to the sending of $m_{i+1}$.
It is not immediately clear how a ``message chain" should be defined in the presence of resetting faults: should a message sent after a reset be counted as continuing a chain of messages received before the reset?
Since our strongly adaptive adversary is constrained to keep to schedules that are approximately synchronous, we will employ a more obvious measure, namely the number of acceptable windows that pass before the first processor decides. When we later reformulate our techniques to obtain a lower bound for a class of algorithms in the presence of crash failures, we will define the running time of an execution as the length of the longest message chain preceding a decision.

\section{Feasibility Against the Strongly Adaptive Adversary}

Both Ben-Or \cite{BO83} and Bracha \cite{B84} provide expected exponential time algorithms for Byzantine agreement against a full-information asynchronous adversary (terminating and succeeding with probability 1). Bracha's algorithm introduces a bit more complexity in order to achieve the optimal resilience of $t < \frac{n}{3}$ in the Byzantine setting.

Inspired by these algorithms, we provide a close variant that succeeds against the strongly adaptive adversary. We will not be concerned with obtaining the optimal resilience, and so will favor simplicity of presentation over possible improvements to the constant fraction of resets allowed per acceptable window. The algorithm is parameterized by several thresholds, $T_1\geq T_2 \geq T_3$, and we will discuss appropriate settings of these below.

Throughout the algorithm, each processor $p$ will store its input bit, its (write-once) output bit, and a few additional variables. The variable $r_p$ will hold the current ``round number" and is initialized to 1. The variable $x_p$ is initialized to be equal to the input bit of processor $p$.

\paragraph{step 1:} Send the message $(r_p,x_p)$ to all processors.

\paragraph{step 2:} Wait until $T_1$ messages of type $(r_q, x_q)$ have arrived from other processors with values of $r_q = r_p$.

\paragraph{step 3:} If at least $T_2$ of these $T_1$ messages have the same bit value $v$ for the last entry, then write $v$ to the output bit (assuming this bit is not yet written). If at least $T_3$ of these $T_1$ messages have the same bit $v$, then set $x_p = v$. Otherwise, set $x_p$ to be a freshly sampled uniformly random bit.

\paragraph{step 4:} Set $r_p =r_p +1$ and return to step 1.

\paragraph{handling resets}
To address resets, a processor $p$ also does the following. If $p$ has just been reset (an event that is detectable), then processor $p$
waits to receive at least $T_1$ messages $(r_q, x_q)$ from other processors with a common value of $r$. It then sets its own $r_p$ value to the match this and returns to step 3 above (note that a newly reset processor refrains from sending messages until it resumes normal operation).

We note that $2T_3 > T_1$ should hold in order for the behavior in step 3 to be clear. This is a constraint we will always impose.

\begin{theorem}The above algorithm achieves measure one correctness and termination against the strongly adaptive adversary for $t < \frac{n}{6}$ when the thresholds $T_1, T_2,  T_3$ are set to satisfy $n-2t \geq T_1 \geq T_2 \geq T_3 +t$, and $2T_3 > n$.
\end{theorem}

\begin{proof} We note that whenever $t < \frac{n}{6}$, the constraints on the thresholds $T_1, T_2, T_3$ specified above are achievable by setting $T_1 := n-2t$, $T_2 = T_1$, $T_3 = n-3t$. (Having a smaller value of $t$ allows one to set $T_2$ smaller than $T_1$, which will lead to improvement in running time but is not relevant for measure one correctness and termination.)

We first establish measure one correctness. Suppose that $\sigma$ is a reachable configuration in which some processor has decided. We consider a non-zero probability partial execution composed of acceptable windows leading to $\sigma$ (such a partial execution must exist by definition of reachability). Now, we let $w$ denote the earliest acceptable window in which a decision is made, and we let $p$ denote a processor deciding in window $w$. In order to decide on a bit $v$, $p$ must have received $\geq T_2$ messages of the form $(r, v)$ for its current value of $r_p$.
Each other processor $q$ must have received $\geq T_2 - t$ of these messages $(r,v)$ during the receiving steps in window $w$.

We now consider how the internal round numbers $r$ maintained by the processors evolve during a sequence of acceptable windows. Initially, all round numbers are 1. Assuming that $n-t \geq T_1$, all processors will increment $r$ to be 2 during the first acceptable window. Thus, entering the second acceptable window, all processors that were \emph{not reset} during the first window will have $r =2$, and the reset processors will have blank $r$ values (denoted by $\bot$). During the second acceptable window, each processor will receive at least $n-2t$ messages from processors with $r$ values equal to 2. Assuming that $T_1 \leq n-2t$, every processor will then have $r =3$ before the resetting steps.

Extending this reasoning inductively, we see that in window $w$, at least $n-t$ processors will enter the window with $r = w$ (with the rest having $r = \bot$). Again assuming that $T_1 \leq n - 2t$ and additionally assuming that $T_2 -t \geq T_3$, every processor will have $x_q =v$ and $r = w+1$ just before the resetting steps that conclude window $w$. It follows that every processor who has not yet decided will decide $v$ in window $w+1$. We must also check that it is impossible for some processor to decide the opposite of $v$ during window $w$. This is impossible as long as $2T_2 >n$.

We have thus shown that it is impossible to obtain contradicting decision values in a reachable configuration. To see that decision values conflicting with a unanimous setting of the inputs are also impossible, note that if all inputs are equal to a common value $v$, then all processors will decide $v$ in the first acceptable window. This completes our proof of measure one correctness.

To establish measure one termination, we first argue that during any given acceptable window, no two processors $p$ and $q$ can fix $x_p$ and $x_q$ deterministically to conflicting values. If $p$ deterministically sets $x_p$ to $v$, this means it received $\geq T_3$ messages with the value $v$. Processor $q$ could not have received $ \geq T_3$ messages with the opposite value if we impose the constraint that $2T_3 > n$. Assuming this constraint, no two processors can deterministically set conflicting values. Thus, there is at least a $2^{-n}$ probability that all processors $p$ set the same value for $x_p$ during any given window. Thus, the probability of not deciding approaches 0 as the number of acceptable windows approaches infinity. This implies measure one termination.
\end{proof}

We observe that for any constant $c < \frac{1}{6}$, setting $t = cn$ makes measure one correctness and termination attainable through the above algorithm, but the algorithm will incur exponential running time (with high probability) against an adversary that chooses initial inputs evenly split between 0 and 1. To see this, note that $T_3$ will always need to be $> \frac{1}{2}n$, and hence $T_2$ will always need to be $ > (\frac{1}{2} +c) n$. Decision will then be contingent on obtaining a strong majority that occurs with probability that is exponentially small (depending on $c$). This is a consequence of the simple fact that with high probability, a sampling of $n$ independent uniformly random bits will yield a deviation of only $\mathcal{O} (n^{\frac{1}{2}+\epsilon})$ from the mean (for any small $\epsilon >0$). Hence, with high probability per round, the adversary can continually extend the execution to last one more round without deciding by showing every processor an approximate split between 0 and 1 messages, and then having all of them set their next bits randomly in step 3. We expect this to continue for an exponential number of rounds until a strong enough majority happens by chance to prevent the adversary from continuing in this fashion.

\section{Impossibility of Expected Polynomial Time Against the Strongly Adaptive Adversary}\label{sec:lowerbound}
Here we establish our main result: an exponential lower bound on the running time for any algorithm achieving measure one correctness and termination against the strongly adaptive adversary.

\begin{theorem}\label{thm:main} We set $t = cn$ where $c>0$ is any fixed positive constant. Then there exist absolute positive constants $C,\alpha$ (depending only on $c$) such that, for any algorithm achieving measure one correctness and termination, there is a strongly adaptive adversary and a setting of the inputs bits such that with probability $\geq \frac{1}{2}$, the running time is $\geq C e^{\alpha n}$.
\end{theorem}

We first give a high-level outline of our proof. As a base case, we observe that the reachable configurations corresponding to a decision of 0 and the reachable configurations corresponding to a decision of 1 form two sets (denoted $Z_0^0$ and $Z_1^0$) that are significantly separated in Hamming distance. Intuitively, if conflicting decision states were too close, then the differing processors could be temporarily silenced, and the other processors could be forced to make a decision that could conflict. By interpolating over the input possibilities and applying Talagrand's inequality, we could use this base observation to prove that there is a setting of the inputs such that reaching a decision in just one ``round" of communication is very unlikely.

To work up to analyzing many rounds, we inductively build pairs of sets $Z_0^k$ and $Z_1^k$ of configurations further out from decisions. These pairs of sets will remain Hamming-separated and will be designed so that if a configuration is \emph{not} in $Z_0^k$, say, then the adversary will have a good chance of continuing the execution for $k$ more rounds without a decision of 0 occurring.
We define $Z_0^k$ such that if we start from a configuration in $Z_0^k$ and apply certain acceptable windows, then there is always a sufficient chance of landing in $Z_0^{k-1}$. We define $Z_1^k$ analogously. To show that $Z_0^k$ and $Z_1^k$ are still significantly separated in Hamming distance, we argue that if they were too close, this would imply the existence of a single product distribution placing too much weight simultaneously on $Z_0^{k-1}$ and $Z_1^{k-1}$. Since these are assumed to be Hamming-separated by the inductive hypothesis, this would contradict Talagrand's inequality.

%

We then show that as long as one avoids the union of the sets $Z_0^k$ and $Z_1^k$, then there is an acceptable window that can be used to extend the execution to a state avoiding $Z_0^{k-1} \cup Z_1^{k-1}$ with high probability. This is essentially an interpolation argument: since we know there is a choice of window that avoids $Z_0^{k-1}$ and another that avoids $Z_1^{k-1}$, we can interpolate to obtain a single choice of extending window that avoids the union.
Finally, we show that if one begins outside of $Z_0^k \cup Z_1^k$, then one can extend the execution for $k$ steps without a decision occurring with constant probability for a value of $k$ that is exponential in $n$.

To prove this theorem formally, we develop some key lemmas and definitions in the next subsections.

\subsection{A Probabilistic Lemma}
We will crucially rely on Talagrand's inequality, a very general tool for studying product measures. We will state a consequence of it here in the context of Hamming distance, as we will not need the additional generality provided by the full statement. A fuller statement and proof can be found in \cite{AS}, for example.

We first develop some convenient notation. We let $\Omega = \prod_{i=1}^n \Omega_i$, where each $\Omega_i$ is a probability space and $\Omega$ is endowed with the product measure.
We employ the usual notion of Hamming distance between points in $\Omega$: for $x = (x_1, \ldots, x_n)$ and $y = (y_1, \ldots, y_n) \in \Omega$, we define $\Delta(x,y)$ to the number of coordinates $i$ such that $x_i \neq y_i$.
Given a set $A \subseteq \Omega$ and a point $x = (x_1, \ldots, x_n) \in \Omega$, we define the Hamming distance $\Delta(x,A)$ between the point $x$ and the set $A$ to be the minimal Hamming distance attained between $x$ and a point $a \in A$:

\begin{definition} For $A \subseteq \Omega$ and $x \in \Omega$,
\[ \Delta(x,A) := \min \{ \Delta(x,a) : a \in A\}.\]
\end{definition}

Similarly, we define the Hamming distance between two sets $A, B \subseteq \Omega$ to be the minimal Hamming distance attained between a point $a \in A$ and a point $b \in B$:

\begin{definition} For $A, B \subseteq \Omega$,
\[ \Delta(A, B) := \min \{ \Delta(a,b) : a \in A, b \in B\}.\]
\end{definition}

Finally, given a set $A \subseteq \Omega$ and a non-negative real number $d$, we define the set $\mathcal{B}(A,d)$ to be the subset of points in $\Omega$ which are at a Hamming distance of at most $d$ from $A$:

\begin{definition} For $A \subseteq \Omega$ and $d \geq 0$,
\[ \mathcal{B}(A,d) := \{x \in \Omega: \Delta(x,A) \leq d\}.\]
\end{definition}

We are now prepared to state the required consequence of Talagrand's inequality (see \cite{AS}, for example):
\begin{lemma}\label{lem:Talagrand} For any $A \subseteq \Omega$ and any $d \geq 0$,
\[ \mathbb{P} [A] (1 - \mathbb{P}[\mathcal{B}(A,d)]) \leq e^{-\frac{d^2}{4n}}.\]
\end{lemma}

\subsection{The Building Blocks of the Proof}\label{sec:bb}
We now recursively define two sequences of subsets of $\Sigma^n$ that will form the building blocks of our proof of Theorem \ref{thm:main}. These definitions will be made with respect to a fixed algorithm $\mathcal{A}$ and a threshold parameter $\tau>0$ that we will set later.
Our base sets are defined as follows:

\begin{definition} We let $Z_0^0$ denote the set of reachable configurations in $\Sigma^n$ such that at least one processor has written 0 to its output bit. Similarly, we let $Z_1^0$ denote the set of reachable configurations in $\Sigma^n$ such that at least one processor has written 1 to its output bit.
\end{definition}

\begin{lemma}\label{lem:base} If the algorithm $\mathcal{A}$ satisfies measure one correctness and measure one termination, then
$\Delta\left(Z_0^0, Z_1^0\right) > t$.
\end{lemma}

\begin{proof} We suppose not. Then there exist reachable configurations $\sigma, \gamma \in \Sigma^n$ such that $\sigma \in Z_0^0$, $\gamma \in Z_1^0$, and $\Delta(\sigma,\gamma) \leq t$. Without loss of generality, we suppose that $\sigma$ and $\gamma$ only differ in the first $t$ coordinates (i.e. only in the local states of processors 1 through $t$).
Consider a non-zero probability partial execution composed of acceptable windows that results in configuration $\sigma$. The adversary can continue such an execution by always delivering the messages from the last $n-t$ processors. This will allow an arbitrarily long sequence of new acceptable windows.

Since the algorithm $\mathcal{A}$ satisfies measure one termination, if the adversary keeps extending this execution by appending new acceptable windows, with probability one a decision must eventually be reached, and since $\sigma \in Z_0^0$, this decision must be 0 (with probability one). However, we can apply the same argument to a partial execution reaching $\gamma$ and then similarly delivering messages only from the last $n-t$ processors. Since the distribution of the states of the last $n-t$ processors is the same in both cases, it must be that their decision is also the same. Since $\gamma \in Z_1^0$, this contradicts measure one correctness.
\end{proof}

Given sets $R, S_1, \ldots, S_n \subseteq [n]$ satisfying $|R| \leq t$, $|S_i| \geq n-t \; \forall i$, we say the strongly adaptive adversary can \emph{apply} this set to a reachable configuration $\sigma \in \Sigma^n$, meaning that the adversary can execute sending steps for all processors, deliver to each processor $i$ the messages sent to it by senders in $S_i$, in some fixed order, and then reset the processors in $R$. Note that the application of sets $R, S_1, \ldots, S_n$ with the specified properties results in an acceptable window (by definition).

Once we have defined sets $Z_0^{k-1}$ and $Z_1^{k-1}$ for some positive integer $k$, we define the next sets $Z_0^{k}$ and $Z_1^{k}$ as follows:

\begin{definition} We let $Z_0^k$ denote the set of reachable configurations in $\Sigma^n$ such that, for any sets $R$, $S$ such that $|R| \leq t$, $|S|\geq n-t$, the adversary applying $R,S, S, \ldots, S$ to the configuration will result in a new configuration that belongs to $Z_0^{k-1}$ with probability $> \tau$. Similarly, we let $Z_1^k$ denote the set of reachable configurations in $\Sigma^n$ such that, for such $R$ and $S$, the adversary applying $R,S, S, \ldots, S$ to the configuration will result in a new configuration that belongs to $Z_1^{k-1}$ with probability $> \tau$.
\end{definition}

\begin{lemma}\label{lem:induct} If the algorithm $\mathcal{A}$ satisfies measure one correctness and termination and $\tau \geq e^{-\frac{t^2}{8n}}$, then $\Delta\left(Z_0^k, Z_1^k\right) > t$ for all non-negative integers $k$.
\end{lemma}

\begin{proof}
We proceed by induction on $k$. The base case $k=0$ is addressed above in Lemma \ref{lem:base}.
We assume the result holds for $k-1 \geq 0$, and we suppose it is false for $k$.
Then there exist reachable configurations $\sigma, \gamma \in \Sigma^n$ such that $\sigma \in Z_0^k$, $\gamma \in Z_1^k$, and $\Delta(\sigma,\gamma) \leq t$. Without loss of generality, we suppose that $\sigma$ and $\gamma$ only differ in the first $t$ coordinates. We let $R$ denote the set $\{1, 2, \ldots, t\}$ and $S$ denote the set $\{t+1, t+2, \ldots, n\}$. By definition of $Z_0^k$, if the adversary applies $R, S, \ldots, S$ to $\sigma$, this will with probability $> \tau$ result in a new configuration belonging to $Z_0^{k-1}$. Similarly, if the adversary applies $R, S, \ldots, S$ to $\gamma$, this will with probability $> \tau$ result in a new configuration belonging to $Z_1^{k-1}$.

We first suppose that both $\sigma$ and $\gamma$ are configurations in which no decisions have occurred. In other words, no processors have yet written to their output bits. Assuming this, the resets will obliterate the differences between the first $t$ processor states. Hence
the distribution of the resulting configuration is \emph{identical} in these two cases, as it is independent of the prior contents of the memories of the reset processors, as these have been erased (and their messages went undelivered). Since local randomness is sampled independently by each processor only in response to the message receipts in the window, which occur \emph{after} the deterministic sending steps, the distribution on the resulting configuration (which is reachable with probability one) is in fact a product distribution. This distribution places weight $> \tau$ on each of two sets, $Z_0^{k-1}$ and $Z_1^{k-1}$, that are separated by a Hamming distance $ > t$.
Applying Lemma \ref{lem:Talagrand}, we thus have that
\[\tau^2 < e^{-\frac{t^2}{4n}} \Leftrightarrow \tau  < e^{-\frac{t^2}{8n}}.\]
This contradicts our stipulation on the value of $\tau$.

We finally consider the case where some decision has already been made in $\sigma$. Since $\sigma \in Z_0^k$, this decision must be 0. However, repeatedly applying acceptable windows of $R, S, \ldots, S$ to $\sigma$ must result in a decision of 1 with nonzero probability, since $\gamma \in Z_1^k$, and the distribution of the final $n-t$ states here is independent of the output bits of the first $t$ processors, as their messages are never delivered. This contradicts measure one correctness.
\end{proof}

We now prove that if a reachable configuration is not in $Z_0^k$ or $Z_1^k$, then the adversary can choose the next acceptable window in a way that will (with high probability) avoid landing in $Z_0^{k-1} \cup Z_1^{k-1}$. The intuition for this can be developed as follows. We know that there is a product distribution induced by an acceptable window that places low probability on $Z_0^{k-1}$, and we know there is a (potentially different) product distribution induced by an acceptable window that places low probability on $Z_1^{k-1}$. We will obtain a single product distribution that places low probability on both sets simultaneously by interpolating between these two distributions. We use the fact that Lemma \ref{lem:Talagrand} yields graceful degradation in the quality of the threshold for ``low probability" as we perturb one coordinate of the product distribution at a time. If we interpolate carefully, we can also ensure that the interpolated distribution we obtain is itself induced by an acceptable window.

\begin{lemma}\label{lem:avoid} Suppose the algorithm $\mathcal{A}$ satisfies measure one correctness and termination and $\tau := e^{-\frac{t^2}{8n}}$. Then, for any reachable configuration $\sigma$ not in $Z_0^k \cup Z_1^k$, there exist sets $R, S_1, \ldots, S_n$ that can be applied to $\sigma$ such that the resulting reachable configuration falls outside $Z_0^{k-1} \cup Z_1^{k-1}$ with probability $\geq 1 - 2e^{-\frac{(t-1)^2}{8n}}$.
\end{lemma}

\begin{proof}
Consider a reachable $\sigma$ in the complement of $Z_0^k \cup Z_1^k$. By definition of $Z_0^k$, this means there is some choice of $R, S$ such that applying $R, S, \ldots, S$ to $\sigma$ will avoid $Z_0^{k-1}$ with probability $\geq 1 - \tau$. Similarly, by definition of $Z_1^k$, there is some choice $R',S'$ such that applying $R', S', \ldots, S'$ to $\sigma$ will avoid $Z_1^{k-1}$ with probability $\geq 1 - \tau$.

We assume without loss of generality that $R' = \{1, 2, \ldots, t'\}$ for some $t' \leq t$. For each $j$ from 0 to $n$, we define the set $R_j$ to be the union of $R \cap \{1, 2, \ldots, j\}$ and $R' \cap \{j+1, \ldots, t'\}$. We observe that $|R_j| \leq t$ for each $j$. We also define $S^j_i := S$ for $i \leq j$ and $S^j_i := S'$ for $i > j$. Then, for each $j$,  we can apply $R_j, S^j_1, \ldots, S^j_n$ to $\sigma$ to produce a new reachable configuration. For each $j$, this induces a product distribution $\pi_j$ on the set of reachable configurations.

By construction, the distribution $\pi_0$ places probability $\leq \tau$ on $Z_1^{k-1}$ and the distribution $\pi_n$ places probability $\leq \tau$ on $Z_0^{k-1}$. The first $j$ coordinates of $\pi_j$ have the same distributions as in $\pi_n$, will the remaining coordinates have the same distribution as in $\pi_0$. We define $\eta:=e^{-\frac{(t-1)^2}{8n}}$. We let $j^*$ denote the minimal value of $j$ such that $\pi_j$ places probability $\leq \eta$ on $Z_0^{k-1}$. (Such a $j^*$ exists since $j = n$ satisfies this condition.) If $j^* = 0$, then $\pi_0$ then places probability $\leq \eta$ on \emph{each of} $Z_0^{k-1}$ and $Z_1^{k-1}$. Otherwise, we argue as follows.

We use $\mathbb{P}_{\pi_j}(A)$ for a set $A$ to denote the probability that $\pi_j$ places on a set $A$. We claim that:
\begin{equation}\label{eq:ballexpand}
\mathbb{P}_{\pi_{j^*}}\left[\mathcal{B}\left(Z_0^{k-1},1\right)\right] \geq \mathbb{P}_{\pi_{j^*-1}} \left[ Z_0^{k-1}\right].
\end{equation}
To see this, consider that the product distributions $\pi_{j^*}$ and $\pi_{j^*-1}$ only differ in a single coordinate. Thus, if we sample a configuration according to $\pi_{j^*-1}$ and obtain a result in $Z_0^{k-1}$, we can resample the differing coordinate to match $\pi_{j^*}$ and we are guaranteed to obtain a result in $\mathcal{B}\left(Z_0^{k-1},1\right)$. The inequality (\ref{eq:ballexpand}) follows.

We observe that the set $\mathcal{B}(Z_1^{k-1}, t-1)$ is disjoint from the set $\mathcal{B}(Z_0^{k-1},1)$, since $\Delta(Z_0^{k-1}, Z_1^{k-1}) > t$.
Hence,
\begin{equation}\label{eq:prob}
1 - \mathbb{P}_{\pi_{j^*}}\left[\mathcal{B}\left(Z_1^{k-1},t-1\right)\right] \geq \mathbb{P}_{\pi_{j^*}}\left[\mathcal{B}\left(Z_0^{k-1},1\right)\right].
\end{equation}
Combining (\ref{eq:ballexpand}) and (\ref{eq:prob}), we see that
\[ \mathbb{P}_{\pi_{j^*}}\left[Z_1^{k-1}\right] \left(1 - \mathbb{P}_{\pi_{j^*}}\left[\mathcal{B}\left(Z_1^{k-1},t-1\right)\right]\right) \geq \eta \mathbb{P}_{\pi_{j^*}} \left[Z_1^{k-1}\right].\]
Applying Lemma \ref{lem:Talagrand} and recalling the definition of $\eta$, we have
\[\mathbb{P}_{\pi_{j^*}} \left[Z_1^{k-1}\right] \leq \frac{1}{\eta} e^{-\frac{(t-1)^2}{4n}} = \eta.\]

We now have a product distribution $\pi_{j^*}$ induced by an acceptable window that places probability $\leq \eta$ on each set $Z_0^{k-1}, Z_1^{k-1}$, and hence $\mathbb{P}_{\pi_{j^*}}\left[Z_0^{k-1} \cup Z_1^{k-1}\right] \leq 2 \eta$, as required.
\end{proof}

\subsection{Proof of Theorem \ref{thm:main}}\label{sec:putTogether}

We now employ the notation and lemmas of the previous subsections to prove Theorem \ref{thm:main}.  We define $\alpha := \frac{c^2}{9}$ and we define $C$ sufficiently small such that
\begin{equation}\label{eq:expsteps}
Ce^{\alpha n} \leq \frac{1}{4} e^{\frac{(cn-1)^2}{8n}}
\end{equation}
holds for all positive integers $n$. For convenience of notation, we define $E:= Ce^{\alpha n}$.

We consider the sets $Z_0^E, Z_1^E$. By Lemma \ref{lem:induct}, we know that $\Delta \left(Z_0^E, Z_1^E\right) > t$. We consider an initial configuration $\sigma$ in which all input bits are set to 0. Then, it must be the case that $\sigma \notin Z_1^E$. Otherwise, there would be a non-zero probability partial execution beginning with $\sigma$ and leading to a decision of 1, which contradicts measure one correctness. Similarly, an initial configuration $\gamma$ in which all input bits are set to 1 cannot belong to $Z_0^E$. Hence, as we interpolate between $\sigma$ and $\gamma$, changing the input bit of one processors at a time, we must discover an initial configuration $\delta$ such that $\delta \notin Z_0^E \cup Z_1^E$. We fix this setting of the inputs.

Our strongly adaptive adversary is now defined as follows. Confronted with a partial execution resulting in a configuration $\sigma$, the adversary determines the maximum value of $k\leq E$ such that $\sigma \notin Z_0^k \cup Z_1^k$. If no such $k$ exists, it continues arbitrarily within the constraint of producing acceptable windows. If such a $k$ does exist, then it applies the sequence of sets guaranteed by Lemma \ref{lem:avoid} in order to yield a $\geq 1 - 2 e^{-\frac{(cn-1)^2}{8n}}$ probability of reaching a new configuration at the end of the acceptable window that is not in $Z_0^{k-1}\cup Z_1^{k-1}$.

Since we begin with an initial configuration that is not in $Z_0^E \cup Z_1^E$, the probability that this strongly adaptive adversary will succeed in causing $\geq E$ acceptable windows to occur before any decision is made is at least:
\[1 - 2E e^{-\frac{(cn-1)^2}{8n}} \geq \frac{1}{2},\]
recalling (\ref{eq:expsteps}) and the definition of $E$. This completes our proof of Theorem \ref{thm:main}.

\section{Consequence for Resilience Against Crash Failures}
The techniques developed to prove the exponential lower bound in the previous section have implications beyond the strongly adaptive adversary. In fact, we can use the same techniques (with a few minor modifications) to prove a lower bound for more traditional asynchronous adversaries that applies to a large, natural class of algorithms. In particular, we consider an asynchronous adversary (with unbounded computational power and knowledge of all messages and internal states) that can cause up to $t$ crash failures during an execution as well as controlling the message scheduling. The only constraint on message delivery is that all messages sent must eventually be delivered, if the recipient has not crashed.

We now define the crucial properties of an algorithm that are needed to apply our lower bound techniques in this setting:
\begin{definition}\label{def:algPropF}
We say an algorithm $\mathcal{A}$ is \emph{forgetful} if each message sent by a processor depends only on its input bit as well as messages received and local randomness sampled since the previous sending event.
\end{definition}

Intuitively, this means that processors do not ``remember" prior events that are not reflected by the most recently received messages.
We define one more property of an algorithm that we will require in conjunction with forgetfulness:
\begin{definition}\label{def:algPropFC}
We say an algorithm $\mathcal{A}$ is \emph{fully communicative} if whenever a processor receives the most recently sent messages from $n-t$ processors, it sends a new message to all $n$ processors.
\end{definition}

These properties are both present in the algorithms in \cite{BO83,B84}, and seem natural in the context of crash failures, where one cannot wait for messages from $t$ processors that may have crashed.
We will prove that our exponential lower bound extends to forgetful, fully communicative algorithms against an adversary able to cause $\leq t$ crash failures, making only minor semantic modifications to the proof in Section \ref{sec:lowerbound}. Intuitively, the combination of forgetfulness and full communication mimics the effect of the resetting failures we previously considered. Now processors are retaining old information forever in their state, but they are basing current actions only on ``recent" information, thereby proceeding as if they have forgotten the outdated portions of their internal state.

In this context, we define a \emph{reachable configuration} to be any configuration that occurs with non-zero probability with at most $t$ crash failures (note that we have dropped the notion of acceptable windows). We analogously define measure one correctness and termination for algorithms by requiring that all reachable configurations display only valid combinations of input and output bits and that any infinite execution in which at most $t$ crash failures occur and all other processors take infinitely many sending and receiving steps has probability zero.

We will prove:
\begin{theorem}\label{thm:main2} We set $t = cn$ where $c>0$ is any fixed positive constant. Then there exist absolute positive constants $C,\alpha$ (depending only on $c$) such that, for any fully communicative and forgetful algorithm achieving measure one correctness and termination, there is an asynchronous adversary and a setting of the inputs bits such that with probability $\geq \frac{1}{2}$, the running time is $\geq C e^{\alpha n}$.
\end{theorem}

\subsection{Definitions and Lemmas}
We first adjust our definitions to obtain suitable sets $Z_0^k, Z_1^k$ for this setting. Since there are no longer any resets, we can assume without loss of generality that the local state of a processor includes a log of all messages the processor has received and sent throughout the execution so far. We will define all of our sets $Z_0^k, Z_1^k$ to be subsets of reachable configurations containing no crashed processors. We will rely on the fully communicative nature of the algorithm to additionally restrict to reachable configurations in which all processors are ready to send to all other processors.

Given a reachable configuration $\sigma$ and sets $S_1, \ldots, S_n$ all of size $\geq n-t$, we say the adversary \emph{applies} these sets to $\sigma$ to mean that the adversary executes the following sequence of steps. First, all processors taking sending steps. Then, each processor $i$ receives the messages just sent to it from the processors in $S_i$ (in some fixed order). Note that when the algorithm is fully communicative, beginning from an initial configuration and repeatedly applying such $n$-tuples of sets will result in every processor sending to every other processor in each sending step.

Now, analogously to the definitions in Section \ref{sec:bb}, we define:

\begin{definition} We let $Z_0^0$ denote the set of reachable configurations in $\Sigma^n$ such that at least one processor has written 0 to its output bit. Similarly, we let $Z_1^0$ denote the set of reachable configurations in $\Sigma^n$ such that at least one processor has written 1 to its output bit.
\end{definition}

\begin{definition} For $k \geq 1$, we let $Z_0^k$ denote the set of reachable configurations in $\Sigma^n$ where all processors are poised to send messages to all other processors and for any set $|S|\geq n-t$, the adversary applying $S, S, \ldots, S$ to the configuration will result in a new configuration that belongs to $Z_0^{k-1}$ with probability $> \tau$. Similarly, we let $Z_1^k$ denote the set of reachable configurations in $\Sigma^n$ such that, for any $|S| \geq n-t$, the adversary applying $S, S, \ldots, S$ to the configuration will result in a new configuration that belongs to $Z_1^{k-1}$ with probability $> \tau$.
\end{definition}

We then have:
\begin{lemma}\label{lem:induct2} If a fully communicative and forgetful algorithm $\mathcal{A}$ satisfies measure one correctness and termination and $\tau \geq e^{-\frac{t^2}{8n}}$, then $\Delta\left(Z_0^k, Z_1^k\right) > t$ for all non-negative integers $k$.
\end{lemma}

\begin{proof} We first establish the base case, i.e. that $\Delta\left(Z_0^0, Z_1^0 \right) > t$. This is similar to the proof of Lemma \ref{lem:base}.

We suppose there exist reachable configurations $\sigma, \gamma \in \Sigma^n$ such that $\sigma \in Z_0^0$, $\gamma \in Z_1^0$, and $\Delta(\sigma,\gamma) \leq t$. Without loss of generality, we suppose that $\sigma$ and $\gamma$ only differ in the first $t$ coordinates (i.e. only in the local states of processors 1 through $t$). We let $S$ denote the set $\{t+1, \ldots, n\}$.

Consider a non-zero probability partial execution that results in configuration $\sigma$.
We can define another non-zero probability partial execution by executing most of the same steps, but crashing each of the first $t$ processors before they send any messages that are not sent in a partial execution resulting in $\gamma$. In other words, we reach a new configuration $\delta$ that agrees with $\sigma, \gamma$ in the final $n-t$ coordinates, and the steps taken by the first $t$ processors in $\delta$ are precisely the common steps taken by these processors in both $\sigma, \gamma$: at the point where the actions of these processors diverge in $\sigma$ and $\gamma$, the processors are crashed.

Now $\delta$ is reachable, and the adversary can continue a partial execution from $\delta$ by continually executing sending and receiving steps among the final $n-t$ processors. Since the algorithm $\mathcal{A}$ satisfies measure one termination, if the adversary keeps extending this execution, with probability one a decision must eventually be reached. Let's suppose that this decision is 1 with non-zero probability. Then the same extension can be applied to a partial execution reaching $\sigma$ and this will yield conflicting decisions with non-zero probability, contradicting measure one correctness. Similarly, if the decision reached from extending $\delta$ is 0 with probability 1, then a partial execution reaching $\gamma$ can be extended to yield conflicting decisions. We may conclude that $\Delta(Z_0^0, Z_1^0) > t$.

We now proceed by induction on $k$ (similarly to the proof of Lemma \ref{lem:induct}).
We assume the result holds for $k-1 \geq 0$, and we suppose it is false for $k$.
Then there exist reachable configurations $\sigma, \gamma \in \Sigma^n$ such that $\sigma \in Z_0^k$, $\gamma \in Z_1^k$, and $\Delta(\sigma,\gamma) \leq t$. Without loss of generality, we suppose that $\sigma$ and $\gamma$ only differ in the first $t$ coordinates. We let $S$ denote the set $\{t+1, t+2, \ldots, n\}$. By definition of $Z_0^k$, if the adversary applies $S, \ldots, S$ to $\sigma$, this will with probability $> \tau$ result in a new configuration belonging to $Z_0^{k-1}$. If the adversary applies $S, \ldots, S$ to $\gamma$, this will with probability $> \tau$ result in a new configuration belonging to $Z_1^{k-1}$.

We first consider the case where no output bits have yet been written in $\sigma$ or $\gamma$.
By the forgetful and fully communicative properties of the algorithm, the distributions of the configurations resulting from applying $S, \ldots, S$ to $\sigma$ and to $\gamma$ only differ in portions of the local state that can no longer affect behavior of the processors going forward. This is because the new messages to be sent by all $n$ processors will only depend on the input bits and the newly received $n-t$ messages, not the prior portions of the state that differed between $\gamma$ and $\sigma$. Hence, the product distribution induced by applying $S, \ldots, S$ to $\gamma$ places weight $> \tau$ of each of two sets, $Z_0^{k-1}$ and $Z_1^{k-1}$, that are separated by a Hamming distance $ > t$.
Applying Lemma \ref{lem:Talagrand}, we thus have that
\[\tau^2 < e^{-\frac{t^2}{4n}} \Leftrightarrow \tau  < e^{-\frac{t^2}{8n}}.\]
This contradicts our stipulation on the value of $\tau$.

In the case that an output bit has already been written as 0 in $\sigma$, say, then we reach a contradiction by repeatedly applying $S, \ldots, S$ to $\sigma$. Since $\gamma \in Z_1^k$, there is a nonzero probability that this results in a decision of 1 by processors outside of the first $t$, since these are unaffected by the first $t$ processor states when these processors are no longer heard.
\end{proof}

\begin{lemma}\label{lem:avoid2} Suppose a fully communicative and forgetful algorithm $\mathcal{A}$ satisfies measure one correctness and termination and $\tau := e^{-\frac{t^2}{8n}}$. Then, for any reachable configuration $\sigma$ not in $Z_0^k \cup Z_1^k$, there exist sets $S_1, \ldots, S_n$ that can be applied to $\sigma$ such that the resulting reachable configuration falls outside $Z_0^{k-1} \cup Z_1^{k-1}$ with probability $\geq 1 - 2e^{-\frac{(t-1)^2}{8n}}$.
\end{lemma}

\begin{proof} This is essentially the same as the proof of Lemma \ref{lem:avoid}, but we restate the argument for completeness.
Consider a reachable $\sigma$ in the complement of $Z_0^k \cup Z_1^k$. By definition of $Z_0^k$, this means there is some choice of $S$ such that applying $S, \ldots, S$ to $\sigma$ will avoid $Z_0^{k-1}$ with probability $\geq 1 - \tau$. Similarly, by definition of $Z_1^k$, there is some choice $S'$ such that applying $S', \ldots, S'$ to $\sigma$ will avoid $Z_1^{k-1}$ with probability $\geq 1 - \tau$.

We define $S^j_i := S$ for $i \leq j$ and $S^j_i := S'$ for $i > j$. Then, for each $j$,  we can apply $S^j_1, \ldots, S^j_n$ to $\sigma$ to produce a new reachable configuration. For each $j$, this induces a product distribution $\pi_j$ on the set of reachable configurations.

By construction, the distribution $\pi_0$ places probability $\leq \tau$ on $Z_1^{k-1}$ and the distribution $\pi_n$ places probability $\leq \tau$ on $Z_0^{k-1}$. The first $j$ coordinates of $\pi_j$ have the same distributions as in $\pi_n$, while the remaining coordinates have the same distribution as in $\pi_0$. We define $\eta:=e^{-\frac{(t-1)^2}{8n}}$. We let $j^*$ denote the minimal value of $j$ such that $\pi_j$ places probability $\leq \eta$ on $Z_0^{k-1}$. (Such a $j^*$ exists since $j = n$ satisfies this condition.) If $j^* = 0$, then $\pi_0$ then places probability $\leq \eta$ on \emph{each of} $Z_0^{k-1}$ and $Z_1^{k-1}$. Otherwise, we argue as follows.

We use $\mathbb{P}_{\pi_j}(A)$ for a set $A$ to denote the probability that $\pi_j$ places on a set $A$. We claim that:
\begin{equation}\label{eq:ballexpand2}
\mathbb{P}_{\pi_{j^*}}\left[\mathcal{B}\left(Z_0^{k-1},1\right)\right] \geq \mathbb{P}_{\pi_{j^*-1}} \left[ Z_0^{k-1}\right].
\end{equation}
To see this, consider that the product distributions $\pi_{j^*}$ and $\pi_{j^*-1}$ only differ in a single coordinate. Thus, if we sample a configuration according to $\pi_{j^*-1}$ and obtain a result in $Z_0^{k-1}$, we can resample the differing coordinate to match $\pi_{j^*}$ and we are guaranteed to obtain a result in $\mathcal{B}\left(Z_0^{k-1},1\right)$. The inequality (\ref{eq:ballexpand2}) follows.

We observe that the set $\mathcal{B}(Z_1^{k-1}, t-1)$ is disjoint from the set $\mathcal{B}(Z_0^{k-1},1)$, since $\Delta(Z_0^{k-1}, Z_1^{k-1}) > t$.
Hence,
\begin{equation}\label{eq:prob2}
1 - \mathbb{P}_{\pi_{j^*}}\left[\mathcal{B}\left(Z_1^{k-1},t-1\right)\right] \geq \mathbb{P}_{\pi_{j^*}}\left[\mathcal{B}\left(Z_0^{k-1},1\right)\right].
\end{equation}
Combining (\ref{eq:ballexpand2}) and (\ref{eq:prob2}), we see that
\[ \mathbb{P}_{\pi_{j^*}}\left[Z_1^{k-1}\right] \left(1 - \mathbb{P}_{\pi_{j^*}}\left[\mathcal{B}\left(Z_1^{k-1},t-1\right)\right]\right) \geq \eta \mathbb{P}_{\pi_{j^*}} \left[Z_1^{k-1}\right].\]
Applying Lemma \ref{lem:Talagrand}, we have
\[\mathbb{P}_{\pi_{j^*}} \left[Z_1^{k-1}\right] \leq \frac{1}{\eta} e^{-\frac{(t-1)^2}{4n}} = \eta.\]

We now have a product distribution $\pi_{j^*}$ induced by applying a sequence of sets of size $\geq n-t$ that places probability $\leq \eta$ on each set $Z_0^{k-1}, Z_1^{k-1}$, and hence $\mathbb{P}_{\pi_{j^*}}\left[Z_0^{k-1} \cup Z_1^{k-1}\right] \leq 2 \eta$, as required.
\end{proof}

\subsection{Proof of Theorem \ref{thm:main2}}
We define $\alpha, C, E$ as in Section \ref{sec:putTogether} and consider the sets $Z_0^E, Z_1^E$ as defined above. There must be an initial configuration $\delta$ such that $\delta \notin Z_0^E \cup Z_1^E$. We fix this setting of the inputs.

Beginning with $\delta$, the adversary proceeds as follows. It first applies the sequence of sets guaranteed by Lemma \ref{lem:avoid2} in order to yield a $\geq 1 - 2 e^{-\frac{(cn-1)^2}{8n}}$ probability of reaching a new configuration that is not in $Z_0^{E-1}\cup Z_1^{E-1}$. If it succeeds, it can then apply a new sequence of sets (again furnished by Lemma \ref{lem:avoid2}) in order to yield a $\geq 1 - 2e^{-\frac{cn-1)^2}{8n}}$ probability of reaching a new configuration that is not in $Z_0^{E-2}\cup Z_1^{E-2}$, and so on.

The probability that this adversary will succeed in causing $\geq E$ such iterations to occur before any decision is made is at least:
\[1 - 2E e^{-\frac{(cn-1)^2}{8n}} \geq \frac{1}{2},\]
recalling (\ref{eq:expsteps}) and the definition of $E$.
By the fully communicative property, we know that in each iteration, every processor will send a message to every other processor. This guarantees a message chain of length $E$.
This completes our proof of Theorem \ref{thm:main2}.

\bibliographystyle{plain}
\bibliography{BA}

\end{document}